\documentclass{llncs}
\usepackage{llncsdoc}
\usepackage{amssymb}
\usepackage{amsmath}
\usepackage{epigraph}
\usepackage{paralist}
\usepackage{url}
\usepackage{tikz}
\usepackage{color} 
\usepackage{graphics}
\usepackage{enumerate}
\usepackage{microtype}

\usepackage[boxed, ruled,noend,linesnumbered]{algorithm2e}

\usepackage{wrapfig}
\usepackage{cite}

\newcommand{\myparagraph}[1]{\vspace{2pt}\noindent \textbf{#1}}

\newcounter{linenumber}

\def\P{\ensuremath{\mathcal{P}}}

\def\A{\ensuremath{\mathcal{A}}}

\def\I{\ensuremath{\mathcal{I}}}

\def\HSS{\mathit{csize}}

\newcommand{\setcon}{\mathit{setcon}}
\newcommand{\remove}[1]{}

\def\P{\mathcal{P}}

\def\I {\mathcal{I}}

\def\fair{fair}
\def\Fair{Fair}

\newcommand{\ignore}[1]{}

\newcommand{\trnote}[1]{\textbf{TR: #1}}

\newcommand{\power}[1]{\alpha(#1)}

\title{Agreement Functions \\ for Distributed Computing Models}

\author{
Petr Kuznetsov 
\and
Thibault Rieutord\thanks{Supported by ANR project
  DISCMAT, grant agreement ANR-14-CE35-0010-01.}}

\institute{LTCI, T\'el\'ecom ParisTech, Universit\'e Paris Saclay }

\date{}

\begin{document}

\maketitle

\begin{abstract}
The paper proposes a surprisingly simple characterization of
a large class of models of distributed computing, 
via an \emph{agreement function}: for each set of
processes, the function determines the best level of set consensus these
processes can reach.
We show that the task computability of a large class of \emph{{\fair}}
adversaries that includes, in particular \emph{superset-closed} and
\emph{symmetric} one, is precisely captured by
agreement functions. 
\end{abstract}



\section{Introduction}

In general, a model of distributed computing is a set of \emph{runs}, i.e., all
allowed interleavings of \emph{steps} of concurrent processes.
There are multiple ways to define these sets of runs in a tractable way. 

A natural one is based on \emph{failure models} that
describe the assumptions on where and when failures might occur. 
By the conventional assumption of  \emph{uniform} failures, processes fail with equal and
independent probabilities, giving rise to  the classical model of
\emph{$t$-resilience}, where at most $t$ processes may fail in a given run.
The extreme case of $t=n-1$, where~$n$ is the number of processes in
the system, corresponds to the
\emph{wait-free} model.
\remove{
A solution of a distributed task in the wait-free model must ensure
that a process is able to make progress (i.e., obtain an
\emph{output} on its task invocation) in a finite number of its own
steps, regardless of the behavior of other processes.
In this paper we assume that processes communicate by
reading and writing to a shared memory.  
}

The notion of \emph{adversaries}~\cite{DFGT11} generalizes
uniform failure models by defining a set of process subsets, called \emph{live
  sets}, and assuming that in every model run, the set of \emph{correct}, i.e., taking infinitely many steps, processes
must be a live set.
In this paper, we consider adversarial read-write shared memory
models, i.e., sets of runs in which processes communicate via reading
and writing in the shared memory and live sets define which sets of
processes can be correct.  

\remove{
Another way to define a model is to bound  the \emph{concurrency} of
executions: how many processes can be concurrently \emph{active},
i.e., having started their computation but not yet being done with it.    
Assuming that processes communicate via reading and writing in the
shared memory, 
the model of $k$-concurrency, in which at most $k$ processes are
active, can be shown equivalent, in terms of solving
distributed tasks, to the wait-free model in which processes also
have access to $k$-set consensus objects~\cite{GG11-univ}.
}

\remove{
The sets of models defined via adversaries and concurrency intersect.
Indeed, the wait-free model corresponds to the
adversary that consists of all non-empty process subsets and, also, the model
of $n$-concurrency.
In general, however, live sets and active processes are different
creatures, so we cannot represent a $k$-concurrent model
as an adversary, and vice versa.
The two models correspond to sets of runs that cannot always be related by containment.

Can we find a way to relate the \emph{computational power} of the
two model classes?
}

A conventional way to capture the power of a model is to determine its
\emph{task computability}, i.e., the
set of distributed tasks that can be solved in it.
For example, consider the \emph{$0$-resilient} adversary $\A_{0\textrm{-}\textit{res}}$
defined through a single live set $\{p_1,\ldots,p_n\}$: the
adversary says that no process is allowed to fail (by taking only
finitely many steps).
%
It is easy to see that the model is strong enough to
solve \emph{consensus}, and, thus, any task~\cite{Her91}.~\footnote{In
  the ``universal'' task of consensus,
every process has a private \emph{input} value, and is expected to
produce an \emph{output} value, so
that (validity)~every output is an input of some process, 
(agreement)~no two processes produce different output values, and (termination)~every process
taking sufficiently many steps returns.} 

%
%
%
%
%
%
In this paper, we propose a surprisingly simple characterization of
the task computability of a large class of adversarial models 
through \emph{agreement functions}.

An agreement function $\alpha$ maps subsets of processes $\{p_1,\ldots,p_n\}$ to
positive integers in $\{0,1,\ldots,n\}$.  
For each subset $P$,  $\alpha(P)$ determines, intuitively, the level of
\emph{set consensus} that processes in $P$ can reach when no other process is 
active, i.e., the smallest number of distinct input values they can
decide on.

For example, the agreement function of the wait-free 
shared-memory model is $\alpha_{wf}: P \mapsto |P|$ and 
the $t$-resilient model, where at most $t$ processes may fail or not
participate, has $\alpha_{t,\textit{res}}:\; P \mapsto \max(0,|P|-n+t+1)$.
%

The agreement function of an adversary $\A$ can be computed 
using the notion of \emph{set consensus power} of an adversary introduced
in~\cite{GK11}: $\alpha_{\A}(P)=\setcon(\A|_P)$.
Here $\A|_P$ is the \emph{restriction of
  $\A$ to $P$}, i.e., the adversary defined through the live sets of $\A$ that
are subsets of $P$.

To each agreement function $\alpha$, corresponding to an existing model, 
we associate a particular model, the \emph{$\alpha$-model}.
The $\alpha$-model is defined as the set of runs satisfying the
following property: the set $P$ of \emph{participating}
(taking at least one step) processes in a run is 
such that $\alpha(P)\geq 1$ and is such that at most $\alpha(P)-1$
processes take only finitely many steps in it.
An algorithm solves a task in the~$\alpha$-model if 
processes taking infinitely many steps produces an output.


We show that, for the class of \emph{{\fair}} adversaries, agreement functions ``tell it
all'' about task computability: a task is solvable in a {\fair}
adversarial model with agreement function $\alpha$ 
\emph{if and only if} it is solvable in the $\alpha$-model.
%
{\Fair} adversaries include notably the class of superset-closed~\cite{HR10,Kuz12} 
and the class of symmetric~\cite{Tau10} adversaries.
Intuitively, superset-closed adversaries do not anticipate failures of
processes: if $S\in \A$ and~$S\subseteq S'$, then~$S'\in\A$.   
Symmetric adversaries do not depend on processes identifiers: 
if~$S\in\A$, then for every set of processes $S'$ such that $|S'|=|S|$,
we have $S'\in\A$.   

A corollary of our result is a characterization of the $k$-concurrency model~\cite{GG09}.
Here we use the fact that the $k$-concurrency model is equivalent, with respect to task solvability, to
the \emph{$k$-obstruction-freedom}~\cite{GK11}, a symmetric adversary
consisting of live sets of sizes from $1$ to $k$.
Thus, the agreement function
$\alpha_{k\textit{-conc}}:\; P\mapsto \min(|P|,k)$ captures the
$k$-concurrent task
computability.
An alternative characterization of $k$-concurrency via a compact
\emph{affine} task was suggested in~\cite{GHKR16}. 

There are, however, models that are not captured by their agreement
functions.
We give an example of a \emph{non-{\fair}} adversary that solves
strictly more tasks than its \emph{$\alpha$-model}.
Characterizing the class of models that can be captured through their
agreement function is an intriguing open question.

\remove{The first problem to resolve here is whether there exist non-{\fair}
adversaries which are in this class. }  

The rest of the paper is organized as follows. 
Section~\ref{sec:model} gives model definitions.
In Section~\ref{sec:agr}, we formally define the notion of an agreement
function.
In Section~\ref{sec:adapt}, we present an $\alpha$-adaptive
set consensus algorithm, and in
Section~\ref{sec:univer}, we prove a few useful properties of
$\alpha$-models. 
%
%
In Section~\ref{sec:adv}, we present the class of {\fair} adversary,
show that superset-closed and symmetric adversaries are {\fair} and
that {\fair} adversaries are captured by their agreement functions.
In Section~\ref{sec:counter}, we give examples of models that are
\emph{not} captured by agreement functions.
Section~\ref{sec:related} reviews related work, and
Section~\ref{sec:conc} concludes the paper.

\section{Preliminaries}
\label{sec:model}



\myparagraph{Processes, runs, models.}
Let $\Pi$ be a system of $n$ asynchronous processes, $p_1,\ldots,p_n$
that communicate via a shared atomic-snapshot memory~\cite{AADGMS93}. 
The atomic-snapshot (AS) memory is represented as a vector of $n$ shared
variables, where each process is associated with a distinct position
in this vector, and exports two operations: \emph{update} and
\emph{snapshot}. 
An \emph{update} operation performed by $p_i$
replaces position $i$ with a new value and a
\emph{snapshot} operation returns the current state of the vector.  

We assume that processes run the
\emph{full-information} protocol: 
the first value each process writes is its \emph{input value}.
A process then alternates between taking snapshots of the memory 
and writing back the result of its latest snapshot.
A \emph{run} is thus a sequence
of process identifiers stipulating the order in which the processes
take operations: each odd appearance of $i$ 
corresponds to an \emph{update} and each even appearance corresponds
to a \emph{snapshot}. 
A \emph{model} is a set of runs. 

\myparagraph{Failures and participation.}
A process that takes only finitely many steps of the
full-information protocol in a given run is called \emph{faulty}, otherwise it is called
\emph{correct}. A process that took at least one step in a given run
is called \emph{participating} in it.
The set of participating processes in a given run is called its
\emph{participating set}. Note that, since every process writes its input value in its first step, the inputs of
participating processes are eventually known to every process that
takes sufficiently many steps.

\remove{
\myparagraph{Adversaries.}
A process takes only finitely many steps in a given run is called \emph{faulty}, otherwise it is called
\emph{correct}. It is convenient to model patterns in which process failures can occur using the formalism of
\emph{adversaries}~\cite{DFGT11}. An adversary $\A$ is defined as a set of possible correct process subsets.

In this paper, we assume that adversaries are \emph{superset-closed}~\cite{Kuz12}: each superset of a set of an
element of $\A$ also an element of $\A$. Superset-closed adversaries provide an interesting non-uniform
generalization of the classical $t$-resilience condition~\cite{HR10}: the \emph{$t$-resilient} adversary in a
system of $n$ processes consists of all sets of $n-t$ or more processes.

Formally, an adversary $\A$ is a set of subsets of $\Pi$, called \emph{live sets}, $\A\subseteq 2^{\Pi}$,
satisfying $\forall S\in \A$, $\forall S'\subseteq \Pi$, $S\subseteq S'$: $S'\in\A$.

We say that an execution is \emph{$\A$-compliant} if the set of processes that are correct in that execution
belongs to $\A$. Thus, an adversary $\A$ describes a model consisting of $\A$-compliant executions.
}

\myparagraph{Tasks.}
In this paper, we focus on 
distributed \emph{tasks}~\cite{HS99}.
A process invokes a task with an input value and the task returns an output value, so that the inputs and the
outputs across the processes which invoked the task respect the task
specification.
Formally, a \emph{task} is defined through a set $\I$ of input vectors (one input value for each process),
a set $\mathcal{O}$ of output vectors (one output value for each process), and a total relation $\Delta:\I\mapsto 2^{\mathcal{O}}$
associating to each input vector a set of valid output vectors. An input $\bot$ denote a \emph{not
participating} process and an output value $\bot$ denote an
\emph{undecided} process. Check~\cite{HKR14} for more details.

In the task of \emph{$k$-set consensus}, input values are in a set of values $V$ ($|V|\geq k+1$), output values
are in $V$, and for each input vector $I$ and output vector $O$, $(I,O) \in\Delta$ if the set of non-$\bot$
values in $O$ is a subset of values in $I$ of size at most $k$.
The special case of $1$-set consensus is called \emph{consensus}~\cite{FLP85}. 

\myparagraph{Solving a task.} We say that an algorithm $A$ solves a task $T=(\I,\mathcal{O},\Delta)$ in a model $M$ if $A$ ensures that (1) in every run in which processes start with an input vector $I\in\I$, all decided values form a vector $O\in\mathcal{O}$ such that $(I,O)\in\Delta$, and (2) if the run is in $M$, then every correct process decides.

This gives rise to the notion of task solvability, i.e., a task $T$ is solvable in a model $M$ if and only if there exists an algorithm $A$ which solves $T$ in $M$.

\myparagraph{BGG simulation.}
The principal technical tool in this paper is a simulation technique
that we call the \emph{BGG simulation}, after Borowski, Gafni,
Guerraoui, collecting ideas presented in~\cite{BG93a,Gaf09-EBG,GG09,GG11-univ}. 
The technique allows a system of~$n$ processes that communicate via
read-write shared memory and $k$-set consensus objects to
\emph{simulate} a $k$-process system running an arbitrary read-write
algorithm.

In particular, we can use this technique to run an extended BG simulation~\cite{Gaf09-EBG} on top
of these $k$ simulated processes, which gives a simulation of an
arbitrary $k$-concurrent algorithm.
An important feature of the simulation is that it adapts to the
number of currently active simulated processes $a$: if it goes below
$k$ (after some simulated processes complete their computations), the
number of used simulators also becomes $a$.    
We refer to~\cite{GHKR16} for a detailed description of this
simulation algorithm.

\section{Agreement functions}
\label{sec:agr}

\begin{definition}[Agreement function]
The \emph{agreement function} of a model $M$ is 
a function $\alpha: 2^{\Pi} \to \{0,\ldots,n\}$, such that 
for each $P\in 2^{\Pi}$, in the set of runs of $M$ in which no process
in $\Pi\setminus P$ participates,
iterative $\alpha(P)$-set consensus can be solved, 
but $(\alpha(P)-1)$-set consensus cannot.
By convention, if $M$ contains no (infinite) runs with participating set
$P$, then  $\alpha(P)=0$.
\end{definition}
 
\noindent
Intuitively, for each $P$, we consider a model consisting of runs of $M$ in which only
processes in $P$ participate and determine the best level of set
consensus that can be reached in this model,
with $0$ corresponding to a model that consists of \emph{finite} runs only.

Note the agreement function $\alpha$ of a model $M$ is \emph{monotonic}:
$P\subseteq P'$ $\Rightarrow$ $\alpha(P)\leq \alpha(P')\leq|P'|$. 
Indeed, the set of runs of $M$ where 
the processes in $\Pi\setminus P$ do not take any step 
is a subset of the set of runs of $M$ where 
the processes in $\Pi\setminus P'$ do not take any step. 
Moreover, $|P'|$-set consensus is trivially solvable in any model 
by making processes return their own proposal directly.
In this paper, we only consider monotonic functions~$\alpha$. 

\begin{definition}[$\alpha$-model] 
Given a monotonic agreement function $\alpha$, the \emph{$\alpha$-model} is the set of runs in which, 
the participating set $P$ satisfies: 
(1)~$\alpha(P)\geq 1$; and, (2)~at most $\alpha(P)-1$ participating
processes take only finitely many
steps.
\end{definition}

\noindent
We say that a model is \emph{characterized 
by its agreement function} $\alpha$ if and only if 
it solves the same set of task as the $\alpha$-model.

\begin{definition}[$\alpha$-adaptive set consensus]
  The $\alpha$-adaptive set consensus task is an agreement task satisfying the
  {\bf validity} and {\bf termination} properties of consensus 
 and the {\bf $\alpha$-agreement} property:
if at some time $\tau$, $k$ distinct values have been returned, 
then the current participating set $P_\tau$ 
is such that $\alpha(P_\tau)\geq k$.
\end{definition}


\section{Adaptive set consensus}
\label{sec:adapt}

We can easily show that
any model with agreement function $\alpha$ can solve the 
$\alpha$-adaptive set consensus task, i.e., to achieve the best level
of set consensus without this an priori knowledge of the set of processes
that are allowed to participate.

\begin{algorithm}
 \caption{Adaptive set consensus for process $p_i$.\label{Alg:AdaptiveAgreement}}
\SetKwRepeat{Repeat}{Repeat}{Until}%
\textbf{Shared variables}: $R[1,\dots,n] \leftarrow (\bot,\bot)$\;\label{Alg:l:MemoryInitState}
$\mathbf{Local\ variables}$: $\mathit{parts},P \in 2^\Pi,v \leftarrow \mathit{Proposal}, k\in \mathbb{N},r[1,\dots,n] \leftarrow (\bot,\bot)$\;\label{Alg:l:Init}

\vspace{1em}

$R[i].\mathit{update}(v,0)$\;\label{Alg:l:Input}
$r = R.\mathit{snapshot}()$\;\label{Alg:l:InitPart1}
$P =$ set of processes $p_j$ such that $r[j]\neq(\bot,\bot)$\;\label{Alg:l:InitPart2}

\Repeat{$\mathit{parts}=P$\label{Alg:l:ConstantPart}}{\label{Alg:l:loop-begin}
	$ \mathit{parts} := P$\;
	$k :=$ be the greatest integer such that $(-,k)$ is in $r$\;\label{Alg:l:AdoptProp1}
	$v :=$ be any value such that $(v,k)$ is in $r$\;\label{Alg:l:AdoptProp2}
	$v := \alpha(\mathit{parts})$-$\mathit{agreement}(v)$\;\label{Alg:l:Agrrement}
	$R[i].\mathit{update}(v,|\mathit{parts}|)$\;\label{Alg:l:LockProp}

	\vspace{1em}

	$r = R.\mathit{snapshot}()$\;\label{Alg:l:UpdatePart1}
	$P =$ set of processes $p_j$ such that $r[j] \neq (\bot,\bot)$\;\label{Alg:l:UpdatePart2}
}\label{Alg:l:loop-end}
\Return $v$\;\label{Alg:l:Decide}

\end{algorithm}

Let $M$ be a model and let $\alpha$ be its agreement function.
Recall that, by definition, assuming that subsets of $P$ participate,
there exists an algorithm that solves $\alpha(P)$-set consensus, 
let $\alpha(P)$-$\mathit{agreement}$ be such an algorithm.

We describe now an \emph{$\alpha$-adaptive} set consensus
algorithm providing the ``best'' level of consensus available in every
participating set, without prior knowledge of who may participate.
The algorithm adaptively ensures that if the participating set is $P$, then at most
$\alpha(P)$ distinct input values can be decided.

The algorithm is presented in
Algorithm~\ref{Alg:AdaptiveAgreement}. The idea is the following:
every process $p$ writes its input in shared memory
(line~\ref{Alg:l:Input}), and then takes a snapshot to get the current set
$P$ of participating processes (lines~\ref{Alg:l:InitPart1}--\ref{Alg:l:InitPart2}). 

Then processes adopt a value ``locked'' by a process associated to the largest participation 
(lines~\ref{Alg:l:AdoptProp1}--\ref{Alg:l:AdoptProp2}). It uses this new value as proposal for
a simulated agreement solving $\alpha(parts)$ (feasible as it is accessed only by processes in $parts$ and   $\alpha$ only increases with participation). The decision value obtained is then ``locked'' in memory by writing the value together with the current participation estimation size, $|parts|$, to the shared
memory (line~\ref{Alg:l:LockProp}). If after locking the value, the updated
participating set $P$ (lines~\ref{Alg:l:UpdatePart1}--\ref{Alg:l:UpdatePart2}), has not
changed (line~\ref{Alg:l:ConstantPart}), the process returns its current decision estimate, $v$ (line~\ref{Alg:l:Decide}). Otherwise, if participation has changed, then the same construction is applied again and again until it has been executed with an observed stable participation (lines~\ref{Alg:l:loop-begin}--\ref{Alg:l:loop-end}).

\begin{theorem}
\label{thm:AdaptiveAgreement}
In any run with a participating set $P$, Algorithm~\ref{Alg:AdaptiveAgreement} satisfies the following
properties:
\begin{itemize}
\item \emph{Termination:} All \emph{correct} processes eventually decide.
\item \emph{$\alpha$-Agreement:} if at some time $\tau$, $k$ distinct values have been returned, 
then the current participating set $P_\tau$ 
is such that $\alpha(P_\tau)\geq k$.
\item \emph{Validity:} Each decided value has been proposed by some process.
\end{itemize}
\end{theorem}

\begin{proof}
Let us show that Algorithm~\ref{Alg:AdaptiveAgreement} satisfies the following
properties:
\begin{itemize}
\item  \textbf{Validity:} 
Processes can only decide on their current estimated decision value~$v$ 
(line~\ref{Alg:l:Decide}). 
This value~$v$ is first initialized to their own input proposal (line~\ref{Alg:l:MemoryInitState}), 
and it can then only be replaced by adopting another process's current, 
initialized, estimated decision value. Adaopting this value can either
be done directly (line~\ref{Alg:l:AdoptProp2}), 
or through an agreement protocol (line~\ref{Alg:l:Agrrement})
satisfying the same \emph{validity} property. 
\item \textbf{Termination:} 
Assume that a correct process never decides. 
Then it must be blocked executing the \emph{while loop} indefinitely often 
(at lines~\ref{Alg:l:loop-begin}--\ref{Alg:l:loop-end}). 
So the participation observed must have changed at each iteration 
(line~\ref{Alg:l:ConstantPart}). But, as participation can only increase, 
it must have been strictly increased infinitely often 
-- a contradiction with the finiteness of the system size.
\item \textbf{$\alpha$-Agreement:}
%
We say that a process \emph{returns at level} $t$ if 
it exited the while loop (line~\ref{Alg:l:ConstantPart}) with
a last participation observed of size $t$. Let $l$ be the 
smallest level at which a process \emph{returns}, and let 
$O_l$ the set of values ever written to $R$ at \emph{level} $l$
, i.e., values $v$ such that $(v,l)$ ever appears in
$R$. We shall show that for all $l'>l$, $O_{l'}\subseteq O_l$.

Indeed, let $q$ be the first process to write a value $(v',l')$ (in
line~\ref{Alg:l:LockProp}), such that~$l'>l$, in $R$. Thus, the
immediately preceding snapshot, taken before this write in
lines~\ref{Alg:l:InitPart1} or~\ref{Alg:l:UpdatePart1}, witnessed a
participating set of size $l'$. Hence, the snapshot of $q$ succeeds
the last snapshot (of size $l<l'$) taken by any process $p$ that
\emph{returned at level} $l$. But immediately before taking this last
snapshot, every such process has written $(v,l)$ in R
(line~\ref{Alg:l:LockProp}) for some $v\in O_l$. Therefore $q$ must see 
$(v,l)$ in its snapshot of size $l'$ and, since, by assumption, the snapshot
contains no values written at levels higher than $l$, $q$ must have adopted
some value written at level $l$, i.e., $v'\in O_l$. 
Inductively, we can derive that any subsequent process will have to adopt a value in $O_l$.

We have shown that every returned value must appear in $O_l$, where $l$ is the
smallest \emph{level} at which some process exits the while loop (line~\ref{Alg:l:ConstantPart}).

Participation can only grow due to snapshots inclusion property, 
thus there is a single participating set
of size $l$, $P_l$, which can be observed in a given run.
 A value in $O_l$ has been adopted as the decision returned by
the $\power{P_l}$-agreement, that can return at most $\power{P_l}$
values as it was only accessed by processes in~$P_l$. 
Therefore the set of returned decisions, as included in $O_l$, 
is smaller or equal to~$\power{P_l}$. Last, it easy to see that
for any process returning at a time $\tau$, the participation $P_\tau$ 
at time $\tau$ is such that $P_l\subseteq P_\tau$ and thus that $\power{P_l}\leq\power{P_\tau}$. 
Therefore, the number of distinct returned decisions returned 
at any time $\tau$ is smaller or equal to $\power{P_\tau}$.
\end{itemize}
\end{proof}

This adaptive agreement is simple but central for all our model simulations. All reductions are constructed with a common simulation structure. Processes write their input, then use their access to shared memory and their ability to solve this adaptive agreement to simulate an adaptive BGG simulation where the number of active BG simulators is equal at any time $\tau$ to at most $\alpha_M(P_\tau)$, with $P_\tau$ the participating set at time $\tau$, with at least one of these BG simulators taking infinitely many steps.

\section{Properties of the $\alpha$-model}
\label{sec:univer}

We now relate task solvability in the $\alpha$-model and in $M$.   
More precisely, we show that (1) the agreement function of
the $\alpha$-model is $\alpha$ and
(2) any task $T$ solvable in the $\alpha$-model is also solvable in every
model with agreement function $\alpha$. 

\begin{theorem}
\label{lem:activeResilency}
The agreement function of the $\alpha$-model is $\alpha$.
\end{theorem}

\begin{proof}
Take $P$ such that $\alpha(P)>1$ and consider the set of runs of the
$\alpha$-model in which no process in $\Pi\setminus P$ participates
and, thus, according to the monotonicity property, at most $\alpha(P)-1$ processes are faulty.
To solve $\alpha(P)$-set consensus, we use the \emph{safe-agreement} protocol~\cite{BG93b}, 
the crucial element of BG simulation.
Safe agreement solves consensus if every process that participates in
it takes enough steps. 
The failure of a process then may \emph{block} the safe-agreement protocol. 
In our case as at most $\alpha(P)-1$ processes in $P$ can fail, so we
can simply run $\alpha(P)$ safe agreement protocols: every process
goes through the protocols one by one using its input as a proposed
value, if the protocol blocks, it proceeds to the next one in the
round-robin manner.
The first protocol that returns gives the output value.
Since at most $\alpha(P)-1$ processes are faulty, at least one safe
agreement eventually terminates, and there are at most $\alpha(P)$
distinct outputs.
To see that $(\alpha(P)-1)$ cannot be solved in this set of runs,
recall that one cannot solve $(\alpha(P)-1)$-set consensus
$(\alpha(P)-1)$-resiliently~\cite{BG93b,HS99,SZ00}. 
\end{proof}

\noindent The following result is instrumental in our
characterizations of \emph{fair} adversaries:

\begin{theorem}\label{thm:adaptiveAbstract}
For any task $T$ solvable in an $\alpha$-model, $T$ is solvable in any read-write shared 
memory model which solves the $\alpha$-adaptive set consensus task.
\end{theorem}

\begin{proof} 
Using $\alpha$-adaptive set consensus 
and read-write shared memory,
we can run $BGG$-simulation so that, 
when the participating set is~$P$,
at most $\alpha(P)$ \emph{BG simulators} are activated 
and at least one is live
(i.e., takes part in infinitely many simulation steps).
Moreover, we make a process provided with a (simulated) task output
to stop proposing simulated steps to BGG simulation. 
Hence, the number of active simulators is also bounded by 
the number of participating processes without an output, 
with at least one live BG simulator if there is 
a correct process without a task output.

These BG simulators are used to simulate an execution of 
a protocol solving~$T$ in the $\alpha$-model.
And so, since any finite run can be extended to 
a valid run of the $\alpha$-model, 
the protocol can only provide valid outputs. 

We make BG simulators execute the \emph{breadth-first} simulation: 
every BG simulator executes an infinite loop consisting of 
(1) updating the estimated participating set $P$, then (2) try to 
execute a simulation step of every process in~$P$, one by one.

Now assume that there exist $k\geq 1$ correct processes that
are never provided with a task output. 
BGG simulation ensure that we eventually have at most
$min(k,\alpha(P))$ active simulators, with at least one live among them. 
Let $s$ be such a live simulator. After every process in $P$ have taken
their first steps, $s$ tries to simulate steps for every process of $P$
infinitely often. A process simulation step can be blocked forever only due
to an active but not live BG simulator\footnote{
Note that the extended BG-simulation provides a mechanism which ensures that
a simulation step is not blocked forever by a no longer active BG simulator.}, 
thus there are at most $min(k,\alpha(P))-1$ simulated processes in $P$
taking only finitely many steps.

As at most $\alpha(P)-1$ processes have a finite number of
simulated steps, the simulated run is a valid run of the $\alpha$-model.
Moreover, as at most $k-1$ processes have a finite number of
simulated steps, there is one process never provided with a task output
simulated as a correct process.
But, a protocol solving a task eventually provides 
task outputs to every correct process --- a contradiction.
 \end{proof}

\noindent
Using Theorem~\ref{thm:AdaptiveAgreement} in combination with Theorem~\ref{thm:adaptiveAbstract},
we derive that:

\begin{corollary}
  \label{cor:adaptive}
  Let $M$ be any model, $\alpha_M$ be its agreement function, and $T$
  be any task that is solvable in the $\alpha_M$-model.
  Then $M$ solves $T$.
\end{corollary}

\section{Characterizing fair adversaries}
\label{sec:adv}

An \emph{adversary} $\A$ is a set of subsets of $\Pi$, called \emph{live sets}, $\A\subseteq 2^{\Pi}$.
An infinite run is \emph{$\A$-compliant} if the set of processes that are correct in that run
belongs to~$\A$. An adversarial $\A$-model is thus defined as the set of
$\A$-compliant runs. 

An adversary is \emph{superset-closed}~\cite{Kuz12} if each 
superset of a live set of~$\A$ is also an element of $\A$, i.e., 
if $\forall S\in \A$, $\forall S'\subseteq \Pi$, $S\subseteq S' \implies S'\in\A$. 
Superset-closed adversaries provide a non-uniform
generalization of the classical \emph{$t$-resilient} adversary
consisting of sets of $n-t$ or more processes.

An adversary $\A$ is a \emph{symmetric} adversary if it does not depend on process
identifiers: $\forall S \in \A$, $\forall S' \subseteq \Pi$, $|S'|=|S|
\implies S'\in\A$. Symmetric adversaries provides another interesting
generalization of the classical $t$-resilience condition and
$k$-obstruction-free progress condition~\cite{GG09} which was previously
formalized by Taubenfeld as its symmetric progress conditions~\cite{Tau10}.

\subsection{Set consensus power}

The notion of the \emph{set consensus power}~\cite{GK10} was originally
proposed to capture the power of adversaries in solving
\emph{colorless} tasks~\cite{BG93a,BGLR01}, i.e., tasks that can be defined by relating \emph{sets}
of inputs and outputs, independently of process identifiers.

\begin{definition}\label{def:scn}
The \emph{set consensus power} of $\A$, denoted by $\setcon(\A)$, is defined as follows:
\begin{itemize}
\item If $\A=\emptyset$, then $\setcon(\A)=0$
\item Otherwise, $\setcon(\A)= \max_{S\in \A} \min_{a\in S}
  \setcon(\A|_{S\setminus\{a\}})  + 1{}.$
  \footnote{$\A|_P$ is the adversary consisting of all live sets of $\A$ that are subsets of~$P$.}
\end{itemize}
\end{definition}
Thus, for a non-empty adversary $\A$, $\setcon(\A)$ is determined as
$\setcon(\A|_{S\setminus\{a\}})+1$ where~$S$ is an element of~$\A$ and~$a$ 
is a process in~$S$  that ``max-minimize'' $\setcon(\A|_{S\setminus\{a\}})$.
Note that for $\A\neq\emptyset$, $\setcon(\A)\geq 1$.

It is shown in~\cite{GK10}  that $\setcon(\A)$ is the smallest
$k$ such that $\A$ can solve $k$-set consensus.

It was previously shown in~\cite{GK11} that for a superset-closed adversary $\A$, 
the set consensus power of $\A$ is equal to $\HSS(\A)$, 
where $\HSS(\A)$ denote the minimal hitting set size of $\A$, 
i.e., a minimal subset of $\Pi$ that intersects with each live set of $\A$. 
Therefore if $\A$ is superset-closed, then $\setcon(\A)=\HSS(\A)$.
For a symmetric adversary $\A$, it can be easily derived from 
the definition of $\setcon$ that 
$\setcon(\A)= |\{k\in\{1,\dots,n\}:\exists S\in\A,|S|=k\}|$. 

\begin{theorem}
\label{th:adv}
The agreement function of adversary $\A$ is
$\alpha_{\A}(P)= \setcon(\A|_P)$.
\end{theorem}
\begin{proof}
An algorithm $A_{P}$ that solves $\alpha_{\A}(P)$-set consensus, 
assuming that the participating set is a subset of $P$, 
is a straightforward generalization of the result of~\cite{GK10}. 
It is shown in~\cite{GK10} that $\setcon(\A)$-set consensus can be solved in $\A$. 
But if we restrict the runs to assume that 
the processes in $\Pi\setminus P$ do not take a single step, 
then the set of possible live sets reduces to $\A|_P$. 
Thus using the agreement algorithm of~\cite{GK10} for the adversary $\A|_P$, 
we obtain a $\setcon(\A|_P)$-set consensus algorithm, 
or equivalently, an $\alpha_{\A}(P)$-set consensus algorithm.
\end{proof}

\noindent
It is immediate from Theorem~\ref{th:adv} that $\A\subseteq \A'$ implies
$\setcon(\A)\leq \setcon(\A')$.

\remove{
\begin{property}
The agreement function of an adversary is regular.
\end{property}

\begin{proof}
Consider sets of processes $P$ and $Q$, and 
assume that the participating set is $P\cup Q$.
Now consider the following protocol:
(1) processes in $P$ solve an $\setcon(\A|_P)$ set-consensus algorithm by 
assuming that processes in $Q\setminus (P\cap Q)$ have failed;
(2) processes in $Q\setminus (P\cap Q)$ return directly their
proposal;
(3) processes also return any value decided by some process.
A process terminate as a process in $Q$ is correct and some
process returns according to (2), or else every process in $Q$
are faulty and thus (1) terminate. Moreover, it is easy to see
that at most $\setcon(\A|_P)$ distinct values can be returned by (1)
and at most $|Q\setminus (P\cap Q)|$ values can be returned by (2).
Thus we can derive that~$\alpha_\A(P\cup Q)\leq \alpha_\A(P)+|Q\setminus (P\cap Q)|$.
\end{proof}
}
\subsection{{\Fair} adversaries}

In this paper we propose a class of adversaries which encompasses both
classical classes of super-set closed and symmetric adversaries.
Informally, an adversary is \emph{{\fair}} 
if its set consensus power does not change if only a 
subset of the processes are participating in an agreement protocol.

More precisely, consider $\A$-compliant runs with participating set $P$ and assume that processes in
$Q\subseteq P$ want to reach agreement \emph{among themselves}: only
these processes propose inputs and are expected to produce outputs.
We can only guarantee outputs to processes in $Q$ when the set of 
correct processes include some process in $Q$, i.e.,
when the current live set intersect with $Q$. 
Thus, the best level of set consensus reachable by $Q$ is
defined the set consensus power of adversary  $\A|_{P,Q}=\{S\in \A|_P, S\cap Q\neq
\emptyset\}$, unless $|Q|<\setcon(\A|_P)$.

\begin{definition}\label{def:fair}[{\Fair} adversary]
An adversary $\A$ is {\fair} if and only if:
\[ \forall P \subseteq \Pi, \forall Q\subseteq P, \setcon(\A|_{P,Q})= min(|Q|,\setcon(\A|_P)){}.\]
\end{definition}

\remove{
  Indeed, if only a subset $Q$ of the current participation $P$ is participating to an agreement protocol, processes need to decide if there is a correct process in $Q$. If there is a correct process in $Q$ then it can be assumed that the valid live-set must include a process from $Q$ and therefore solving the best set consensus possible in $Q$ correspond to solve the best set consensus possible in $\A'=\{S\in \A|_P, S\cap Q\neq \emptyset\}$.

\begin{property}
Not all adversaries are {\fair}.
\end{property}

\begin{proof}
To see that not all adversaries are {\fair}, one can look at the adversary $\A=\{\{p_1\},\{p_2,p_3\},\{p_1,p_2,p_3\}\}$. It is easy to see that $\setcon(\A)=2$, but that $\setcon(\{S\in \A, S\cap \{p_2,p_3\}\neq\emptyset\}= \setcon(\{\{p_2,p_3\},\{p_1,p_2,p_3\}\})= 1$ which is not equal to $min(|\{p_2,p_3\}|,\setcon(\A))= 2$. Therefore $\A$ is an example of a non-{\fair} adversary, i.e., if every-process participate then $\{p_2,p_3\}$ can solve consensus between themselves even if the set consensus power of the all adversary is equal to $2$.
\end{proof}
}


\begin{property}
  \label{prop:nonfair}
  \[\setcon(\A|_{P,Q})\leq \min(|Q|,\setcon(\A|_P))\]
\end{property}
\begin{proof}
  For any $P\subseteq \Pi$ and $Q\subseteq P$,
  $\A|_{P,Q}=\{S\in \A|_P, S\cap Q\neq \emptyset\}$ is a subset of $\A|_P$
  and, thus, $\setcon(\A|_{P,Q}) \leq \setcon(\A|_P)$.
  Moreover,  $\setcon(\A|_{P,Q})\leq |Q|$,
  as $|Q|$-set consensus can be solved in $\{S\in \A|_P, S\cap Q\neq
  \emptyset\}$ as follows: every process waits until some
  process in $Q$ writes its input and decides on it.    
\end{proof}


\begin{theorem}
Any superset-closed adversary is {\fair}.
\end{theorem}

\begin{proof}
Suppose that there exists a superset-closed adversary $\A$ that
is not {\fair}, i.e., by Property~\ref{prop:nonfair}, $\exists P\subseteq \Pi,\exists Q\subseteq P,
\setcon(\{S\in \A|_P, S\cap Q\neq \emptyset\})<min(|Q|,\setcon(\A|_P))$.
Clearly $\A|_P$ and $\A|_{P,Q}$ are
also superset-closed and, thus, $\setcon(\A|_P)=\HSS(\A|_P)$ and
$\setcon(\A|_{P,Q})=\HSS(\A|_{P,Q})$.

Since $\setcon(\A|_{P,Q})<|Q|$, a minimal hitting set $H'$ of
$\A|_{P,Q}$ is such that $|H'|<|Q|$, and therefore there exists a process
$q\in Q$, $q\not\in H'$. Also, since $\setcon(\A|_{P,Q})<\setcon(\A|_P)$, $H'$
is not a hitting set of $\A|_P$. 
Thus, there exists $S\in \A|_P$ such that $S\cap
H'=\emptyset$.
Hence, $(S\cup\{q\})\cap H'=\emptyset$.
Since $\A|_P$ is superset closed, we have $S\cup\{q\}\in \A|_P$ and, since
$q\in Q$, $S\cup\{q\}\in \A|_{P,Q}$.
But $(S\cup\{q\})\cap H'=\emptyset$---a contradiction 
with $H'$ being a hitting set of $\A|_{P,Q}$.
\end{proof}

\begin{theorem}
Any symmetric adversary is {\fair}.
\end{theorem}

\begin{proof}
The set consensus power of a generic adversary $\A$ is defined recursively
through finding $S\in \A$ and $p\in S$ which max-minimize the
set consensus power of $\A|_{S\setminus\{p\}}$.
Let us recall that if $\A\subseteq \A'$ then
$\setcon(\A)\leq \setcon(\A')$.
Therefore, $S$ can always be selected to
be \emph{locally maximal}, i.e., such that there is no live set in $S'\in \A$
with $S\subsetneq S'$.

Suppose by contradiction that $\A$ is symmetric but not {\fair}, i.e.,
by Property~\ref{prop:nonfair},
for some $P\subseteq \Pi$ and $Q\subseteq P$, 
$\setcon(\A|_{P,Q})<min(|Q|,\setcon(\A|_P))$.  
We show that if the property holds for $P$ and $Q$ such that $\A|_{P,Q}\neq\emptyset$ then
it also holds for some $P'\subsetneq P$ and $Q'\subseteq Q$.

First, we observe that $|Q|>1$, otherwise $\setcon(\A|_{P,Q})=0$ and, thus,
we have $\A|_{P,Q}=\emptyset$.

Since $\A$ is symmetric, $\A|_P$ is also symmetric.
Thus, for every $S\in\A|_P$ and $p\in S$ such that
$\setcon(\A|_P)= 1+\setcon(\A|_{S\setminus\{p\}})$, any $S'$
such that $|S'|=|S|$ and for any $p'\in S'$, we also have    
$\setcon(\A|_P)= 1+\setcon(\A|_{S'\setminus\{p'\}})$.
 Since we can always choose $S$ to be a maximal set, we derive that
the equality holds for every maximal set $S$ in $\A|_P$ and every
$p\in S$. 

%
Let us recall that, by the definition of $\setcon$, there exists $L\in \A|_{P,Q}$ and $a\in L$ such that 
$\setcon(\A|_{P,Q})= 1+\setcon((\A|_{P,Q})|_{L\setminus\{a\}})=\setcon(\A|_{L,Q})$.
Since $\A|_P$ is symmetric,
for all $L'$, $|L'|=|L|$ and $L\cap Q\subseteq L'\cap Q$, we have
$\setcon(\A|_{L',Q})\geq\setcon(\A|_{L,Q})$. Indeed, modulo a permutation
of process identifiers, $\A|_{L',Q}$ contains all the live sets of
$\A|_{L,Q}$ plus live sets in $\A|_{L'}$ that overlap with $(L'\cap
Q)\setminus(L\cap Q)$.
Since $\setcon(\A|_{L,Q})=\setcon(\A|_{P,Q})$ and $L'\in \A|_{P,Q}$, we
have $\setcon(\A|_{L',Q})=\setcon(\A|_{L,Q})$.
Therefore, for any $a\in L'$, 
$\setcon(\A|_{L'\setminus\{a\},Q})< \setcon(\A|_{L'\setminus\{a\}})$.

In particular, for $L'$ with $L'\cap Q\in \{L',Q\}$,
$\setcon(\A|_{L',Q})=\setcon(\A|_{L,Q})$.
Note that $L'\nsubseteq Q$, otherwise, $\A|_{L',Q}=\A|_{L'}$ and, thus,
$\setcon(\A|_{L',Q})=\setcon(\A|_{L'})=\setcon(\A|_P)$, contradicting
our assumption.

Thus, let us assume that $Q\subsetneq L'$.
%
%
%
Note that $Q'=Q\setminus\{a\}\subsetneq L'\setminus\{a\}$, and since
$|Q|\geq 2$, $Q'\neq\emptyset$, 
we have $\setcon(\A|_{P',Q'})< \setcon(\A|_{P'})$ for
$P'=L'\setminus\{a\}$ and $Q'\subseteq P'$, $Q'\neq\emptyset$.   
Furthermore, since $\setcon(\A|_{P,Q})<|Q|$, we have
$\setcon(\A|_{P',Q'})<|Q'|$. 

By applying this argument inductively, we end up with a live set
$P$ and $Q\subseteq P$ such that $\setcon(\A|_P)\geq 1$,
$Q\neq\emptyset$ and 
$\setcon(\A|_{P,Q})=0$.
By the definition of $\setcon$, $\A|_P\neq\emptyset$ and 
$\A|_{P,Q}=\emptyset$.
But $\A|_P$ is symmetric and $Q\neq\emptyset$, so for every $S\in \A|_P$, there exists
$S'\in\A|_P$ such that $|S|=|S'|$ and $S'\cap Q \neq \emptyset$,
i.e.,  $\A|_{P,Q}\neq\emptyset$---a contradiction. 
\end{proof}

Note that not all adversaries are fair. 
For example, the adversary 
$\A=\{\{p_1\},\{p_2,p_3\},\{p_1,p_2,p_3\}\}$ is not fair. 
On the other hand, not all fair adversaries are either 
super-set closed or symmetric. For example, the adversary
$\A=2^{\{p_1,p_2,p_3\}}\setminus\{p_1,p_2\}$ is fair but is
neither symmetric not super-set closed. Understanding 
what makes an adversary fair is an interesting challenge.

\subsection{Task computability in {\fair} adversarial models}

In this section, we show that the task computability 
of a {\fair} adversarial $\A$-model is fully grasped by 
its associated agreement function $\alpha_\A$.

\begin{algorithm}
 \caption{Code for BG simulator $s_i$ to simulate adversary $\A$.\label{Alg:AdversarySimulation}}
\SetKwRepeat{Repeat}{Repeat}{Forever}%
\textbf{Shared variables:} $R[1,\dots,\alpha_\A(\Pi)] \leftarrow (\bot,\emptyset) ,P_{MEM}[p_1,\dots,p_n] \leftarrow \bot$\;\label{•}l{Alg:Adv:MemoryInitState}
\textbf{Local variables:} $S_{cur},S_{tmp},P,A,W \in 2^\Pi,p_{cur},p_{tmp} \in \mathbb{N},S_{cur}\leftarrow\emptyset$\;\label{Alg:Adv:Init}
	
\vspace{1em}

\Repeat{}{\label{Alg:Adv:loop-begin}
	$P = \{p\in \Pi, P_{MEM}[p]\neq \bot\}$\;\label{Alg:Adv:Participation}
	$A = \{p\in P, P_{MEM}[p]\neq \top\}$\;
	
	\If{$i \geq min(|A|,\alpha_\A(P))$}{\label{Alg:Adv:TestActive}

	$W = P$\;\label{Alg:Adv:LSselection-begin}
	\For{$j=\alpha_\A(\Pi)$ \textbf{down to} $i+1$}{\label{Alg:Adv:WselectLoop}
		$(p_{tmp},S_{tmp}) \leftarrow R[j]$\;
		\If{$(p_{tmp}\neq\bot)\wedge(S_{tmp}\subseteq W)\wedge((\setcon(\A|_{S_{tmp},A})\geq j))$}
                {\label{Alg:Adv:RecentlyActive}
			$W \leftarrow S_{tmp}\setminus\{p_{tmp}\}$\;\label{Alg:Adv:Wselect}
		}
	}
	
	\vspace{1em}
	
	 \If{$(S_{cur}\not\subseteq W)\vee(\setcon(\A|_{S_{cur},A})< i)$}{\label{Alg:Adv:CheckCur}
	\If{$\exists S\in \A|_W, \setcon(\A|_{S,A})\geq i$}{\label{Alg:Adv:CheckExists}
		$S_{cur} = S\in \A|_W \mathbf{\ such\ that\ } \setcon(\A|_{S,A})\geq i$\;\label{Alg:Adv:SelectNew}
	}
	\lElse{$S_{cur}=  S\in \A|_P$}\label{Alg:Adv:SelectAny}
	$p_{cur} = S_{cur}.first()$\;
	$R[i] \leftarrow (p_{cur},S_{cur})$\;\label{Alg:Adv:LSselection-end}
	}
	
	\vspace{1em}
	
	\If{$(\mathbf{SimulateStep}(p_{cur})=SUCCESS)$}{\label{Alg:Adv:LSsimulation-begin}
		\lIf{$\mathbf{Outputed}(p_{cur})$}{$P_{MEM}[p_{cur}]=\top$}\label{Alg:Adv:SetTerminated}
		$p_{cur} = S_{cur}.next(p_{cur})$\;
	}
	\lElse{
		$\mathbf{AbortStep}(p_{cur})$\label{Alg:Adv:LSsimulation-end}
	}
	}
}\label{Alg:Adv:loop-end}
\end{algorithm}

Using BGG simulation, we show that the $\alpha_\A$-model 
can be used to solve any task $T$ solvable in the $\A$-model.
In the simulation, up to $\alpha(P)$ BG simulators execute the
given algorithm solving $T$, where $P$ is the participating set of the
current run.  
\remove{
To do so, the simulation executes an algorithm solving $T$ in the 
$\A$-model and tries to provide a run such that 
(1) the run is $\A$-compliant, i.e., a run in which the set of processes 
performing infinitely many simulated steps is a live set of $\A$, 
and (2) the set of processes performing infinitely many simulated steps
includes a correct process not provided with a task output.
Succeeding in providing such a run is not possible.
Indeed, an $\A$-compliant run of an algorithm solving a task 
eventually provides outputs to every correct process.
}
We adapt the currently simulated live set to include processes 
not yet provided with a task output, and ensure that the chosen live set is
simulated sufficiently long until some active processes are provided
with outputs of $T$.
The simulation terminates as soon as all correct processes are
provided with outputs.      


The code for BG simulator $b_i\in\{b_1,\dots,b_{\alpha_\A(\Pi)}\}$ is
given in Algorithm~\ref{Alg:AdversarySimulation}. 
It consists of two parts: (1)~selecting a live set to simulate
(lines~\ref{Alg:Adv:LSselection-begin}--\ref{Alg:Adv:LSselection-end}), 
and (2)~simulating processes in the selected live set 
(lines~\ref{Alg:Adv:LSsimulation-begin}--\ref{Alg:Adv:LSsimulation-end}).

\myparagraph{Selecting a live set.}
This is the most involved part.
The idea is to select a participating live set 
$L\subseteq P$ such that: 
(1) the set consensus power of $\A|_{L,A}=\{S\in \A|_L,S\cap
A\neq\emptyset\}$, with $A$ the set of participating processes not yet provided
with a task output, is greater than or equal to the BG simulator identifier $i$;
(2) $L$ is a subset of the live sets currently selected 
by live BG simulators with greater identifiers;
(3) $L$ does not contain the processes currently simulated 
by live BG simulators with greater identifiers.

The live set selection in Algorithm~\ref{Alg:AdversarySimulation} consists in two phases.
First, BG simulators determine a \emph{selection window} $W$, $W\subseteq P$, i.e., 
the largest set of processes which is a subset 
of the live sets selected by live BG simulators with greater identifiers, 
and which excludes the processes currently selected 
by live BG simulators with greater identifiers
(lines~\ref{Alg:Adv:LSselection-begin}--\ref{Alg:Adv:Wselect}). 
This is done iteratively on all BG simulators with 
greater identifiers, from the greatest to the lowest. 
At each iteration, if the targeted BG simulator $b_k$ \emph{appears live},
the current window is restricted to the live set selected by $b_k$, 
but excluding the process selected by $b_k$. 
Determining if $b_k$ appears live is simply done by checking whether, 
with the current simulation status observed,
the live set selected by $b_k$ is \emph{valid}, 
i.e., satisfies conditions~(1),~(2) and~(3) above.

The second phase (lines~\ref{Alg:Adv:CheckCur}--\ref{Alg:Adv:LSselection-end}), 
consists in checking if the currently selected live set is valid (line~\ref{Alg:Adv:CheckCur}). 
If not, the BG simulator 
tries to select a live set $L$ which belongs to the selection window $W$, and hence satisfies (2) and (3),
but also such that the set consensus power of $\A_{L,A}$ is greater than $i$, 
the BG simulator identifier (line~\ref{Alg:Adv:SelectNew}).
If the simulator does not find such a live set, 
it simply selects any available live set (line~\ref{Alg:Adv:SelectAny}).

\myparagraph{Simulating a live set.}
The idea is that, if the selected live set does not change, 
the BG simulator simulates steps of 
every process in its selected live set infinitely often. 
Unlike conventional variations of BG simulations, a BG simulator here
does not skip a blocked process simulation,
instead it aborts and re-tries the same simulation step until it is successful.

Intuitively, this does not obstruct progress because,
in case of a conflict, there are two live BG simulators blocked on the 
same simulation step, but the BG simulator with the smaller identifier
will eventually change its selected live set and release the corresponding process.

\myparagraph{Pseudocode.}
The protocol executed by processes in the $\alpha_\A$-model is the following:
Processes first update their status in $P_{MEM}$ 
by replacing $\bot$ with their initial state. 
Then, processes participate in an $\alpha_{\A}$-adaptive BGG simulation 
(i.e., BGG simulation runs on top of an $\alpha_{\A}$-adaptive
set consensus protocol), 
where BG simulators use Algorithm~\ref{Alg:AdversarySimulation} 
to simulate an algorithm solving a given task~$T$ in the adversarial $\A$-model. 
When a process $p$ observes that $P_{MEM}[p]$ has been set to $\top$
(``termination state''),
it stops to propose simulation steps. 

\myparagraph{Proof of correctness.} 
Let $P_f$ be the participating set of the $\alpha_\A$-model run, 
and let $A_f$ be the set of processes $p\in P_f$ such that $P_{MEM}[p]$ is never set to $\top$.

\begin{lemma}
\label{lem:stable}
  There is a time after which variables $P$ and $A$ in
  Algorithm~\ref{Alg:AdversarySimulation} 
  become constant and equal to $A_f$ and $P_f$ for all live BG simulators. 
\end{lemma}

\begin{proof}
Since $\Pi$ is finite, the set of processes $p$ such that $P_{MEM}[p]\neq \bot$ eventually corresponds to $P_f$ 
as the first step of $p$ is to set $P_{MEM}[p]$ to its initial state and 
$P_{MEM}[p]$ can only be updated to $\top$ afterwards. 
As after $P_{MEM}[p]$ is set to~$\top$, it cannot be set to another value, 
eventually, the set of processes from $P_f$ 
such that $P_{MEM}[p]\neq \top$ is equal to~$A_f$.
Live BG simulators update $P$ and $A$ infinitely often, 
so eventually their values of $P$ and $A$ are equal to $P_f$ and $A_f$ respectively.
\end{proof}

\begin{lemma}
\label{lem:bgsim}
  If $A_f$ contains a correct process, then there is a
  correct BG simulator with an identifier smaller or equal to $min(|A_f|,\alpha_\A(P_f))$.
\end{lemma}

\begin{proof}
In our protocol, eventually only correct processes in $A_f$ are 
proposing BGG simulation steps. Thus eventually, at most 
$|A_f|$ distinct simulations steps are proposed. 
The $\alpha_A$-adaptive set consensus protocol used for BGG simulation 
ensures that at most $\alpha_\A(P_f)$ distinct 
proposed values are decided. But as there is a time after which only processes in $A_f$ 
propose values, eventually, $min(|A_f|,\alpha_\A(P_f))$-set consensus
is solved. Thus BGG simulation ensures that, when this is the case, there is a live 
BG simulator with an identifier smaller or equal to $min(|A_f|,\alpha_\A(P_f))$.
\end{proof}

Suppose that $A_f$ contains a correct process, and let $b_m$ be the
greatest live BG simulator such that $m\leq
min(|A_f|,\alpha_\A(P_f))$ (by Lemma~\ref{lem:bgsim}).
Let $S_i(t)$ denote the value of $S_{cur}$ and
let $p_i(t)$ denote the value of $p_{cur}$ at simulator $b_i$ at time
$t$.
Let also $\tau_f$ be the time after which every active but not live BG simulators
have taken all their steps, and after which $A$ and $P$ have become constant and
equal to $A_f$ and~$P_f$ for every live BG simulator (by Lemma~\ref{lem:stable}).

\begin{lemma}
For every live BG simulator $b_s$, with $s\leq
min(|A_f|,\alpha_\A(P_f))$,
eventually,  $b_s$ cannot fail the test on line~\ref{Alg:Adv:CheckExists}.\label{lem:ValidSetcon}
\end{lemma}

\begin{proof}
Consider a correct BG simulator $b_s$ starting a round after time
$\tau_f$. Let $W_s$ be the value of $W$ at the end of
line~\ref{Alg:Adv:Wselect}. Two cases may arise:
 
 - If $W_s=P_f$, as $\A$ is {\fair}, then
$\setcon(\A|_{W_s,A_f})=\min(|A_f|,\setcon(\A|_{P_f}))$.
Thus, $\setcon(\A|_{W_s,A_f})\geq s$.

 - Otherwise, $W_s$ is set on line~\ref{Alg:Adv:Wselect} to some
$S_{target}\setminus \{p_{target}\}$ at some iteration~$l$, 
with $\setcon(\A|_{S_{target},A_f})\geq l$ for $l>s$.
We have $\setcon(\A|_{W_s,A_f}) =
\setcon((\A|_{S_{target},A})|_{S_{target}\setminus\{p_{target}\}})$
which, by the definition of $\setcon$, 
is greater or equal to $\setcon(\A|_{S_{target},A}) - 1\geq l-1\geq s$,
so we have
$\setcon(\A|_{W_s,A_f})\geq s$.

By the definition of $\setcon$, as $\setcon(\A|_{W_s,A_f})\geq s$, 
there exists $S\subseteq W_s$ such that $\setcon(\A|_{S,A_f})\geq s$.
So, eventually $b_s$ will always succeed the test on line~\ref{Alg:Adv:CheckExists}.
\end{proof}

\begin{lemma}
For every live BG simulator $b_s$, with $s\leq min(|A_f|,\alpha_\A(P_f))$, 
eventually, the value of $W$ computed at the end of iteration $m+1$ (at lines~\ref{Alg:Adv:WselectLoop}--\ref{Alg:Adv:Wselect}) is 
equal to some constant value $W_{m,f}$.\label{lem:StableWindow}
\end{lemma}

\begin{proof}
No BG simulator $b_l$, with $l> m$, executes lines~\ref{Alg:Adv:LSselection-begin}--\ref{Alg:Adv:LSsimulation-end} after time~$\tau_f$. Therefore $R[l]$ is constant after time~$\tau_f$, $\forall l>m$. As the computation of~$W$, on lines~\ref{Alg:Adv:LSselection-begin}--\ref{Alg:Adv:Wselect}, only depends on the value of $A$, $P$ and $R[l]$, for $\alpha_\A(\Pi)\geq l>m$, all constant after time~$\tau_f$, then the value of $W$ computed at the end of line~\ref{Alg:Adv:Wselect} for iteration $m+1$ 
is the same at every round initiated after time~$\tau_f$ for any 
live BG simulator $b_s$, with $s\leq min(|A_f|,\alpha_\A(P_f))$.
\end{proof}

\begin{lemma}If $A_f$ contains a correct process,
  then the set of processes with an infinite number of simulated steps
  is a live set of $\A$ containing a process of~$A_f$.\label{lem:ContradictoryLemma}
\end{lemma}

\begin{proof}
As $b_m$ is live, it proceeds to an infinite number of rounds. 
By Lemma~\ref{lem:StableWindow}, eventually $b_m$ computes the same window in every 
round. By Lemma~\ref{lem:ValidSetcon}, if $b_m$ does not have a valid live set 
selected, then it eventually selects a valid one for~$W_{m,f}$. 
Thus, eventually $b_m$ never changes its selected live set. 
Let $S_{m,f}$ be this live set. 
Afterwards, in each round, $b_m$ tries to complete a simulation step of $p_m(t)$ and, 
if successfully completed, changes $p_m(t)$ in a round robin manner among $S_{m,f}$. 
Two cases may arise:

 - If $p_m(t)$ never stabilizes, then the set of processes with 
an infinite number of simulated steps includes $S_{m,f}$. 
By Lemma~\ref{lem:StableWindow}, every other live BG simulator
with a smaller identifier computes the same value of $W$
at the end of round $m+1$ (of the loop at lines~\ref{Alg:Adv:WselectLoop}--\ref{Alg:Adv:Wselect}). 
Thus, after the $S_{m,f}$ is selected by $b_m$, as $S_{m,f}$ is valid, 
every BG simulator will select a subset of $S_{m,f}$ for its window value in every round. 
Moreover, by Lemma~\ref{lem:ValidSetcon}, 
these BG simulators will always find valid live sets to select, 
and so they will eventually simulate only processes in $S_{m,f}$.
Thus, the set of processes with infinitely many simulated steps 
is equal to $S_{m,f}$, a live set intersecting with $A_f$.

 - Otherwise, $p_m(t)$ eventually stabilizes on some $p_{m,f}$. 
Therefore, $b_m$ attempts to complete a simulation step of $p_{m,f}$ infinitely often.
Two sub-cases may arise:
\begin{itemize}
\item Either $|S_{m,f}|=1$ and, therefore, $b_m$ is the only one live
BG simulator performing simulation steps, and thus, 
the set of processes with an infinite number of simulated steps 
is equal to $S_{m,f}$, a live set intersecting with $A_f$.
\item Otherwise, by Lemma~\ref{lem:StableWindow}, 
every live BG simulator with a smaller identifier eventually 
selects a window, and thus a live set (Lemma~\ref{lem:ValidSetcon}),
which is a subset of $S_{m,f}\setminus\{p_{m,f}\}$. Thus 
every live BG simulator with a smaller identifier eventually
selects processes to simulate distinct from $p_{m,f}$ and, thus, 
cannot block $b_m$ infinitely often---a contradiction.
\end{itemize}
\end{proof}

\begin{lemma}
If $\A$ is {\fair}, then any task $T$ solvable in the $\A$-model is solvable in the $\alpha_\A$-model.\label{lem:AdvReduction}
\end{lemma}

\begin{proof}
Let us assume that it is not the case: there exists a task $T$ and a
{\fair} adversary $\A$ such that $T$ is solvable in the adversarial
$\A$-model but not in the $\alpha_\A$-model. As every finite run 
of the $\A$-model can be extended to and $\A$-compliant run, the 
simulated algorithm can only provide valid outputs to the simulated 
processes. Thus, it can only be the case that a correct process is 
not provided with a task output, i.e., belongs to $A_f$. 

Therefore, by Lemma~\ref{lem:ContradictoryLemma}, 
the simulation provides an $\A$-compliant run, i.e.,
the set of processes with an infinite number of simulated steps is a live set.
As the run is $\A$-compliant then each process $p$ 
with an infinite number of simulated steps 
is eventually provided with a task output 
and thus $p_{MEM}[p]$ is set to $\top$. 
Thus, they cannot belong to $A_f$ --- a contradiction.
\end{proof}

Combining Corollary~\ref{cor:adaptive} and Lemma~\ref{lem:AdvReduction} we obtain the following result:
\begin{theorem}
\label{th:adv:task}
For any {\fair} adversary $\A$, the adversarial $\A$-model and the $\alpha_\A$-model are equivalent regarding task solvability.
\end{theorem} 
  
\remove{
\begin{proof}
The ``if'' direction is implied by Lemma~\ref{lem:adaptive}. 
The ``only if'' direction is shown in Lemma~\ref{lem:AdvReduction}.
\end{proof}
}

\section{Agreement functions do not always tell it all}
\label{sec:counter}

We observe that agreement functions are not able to characterize
the task computability power of \emph{all} models. 
In particular there are non-{\fair} adversaries 
not captured by their agreement functions.

Consider for example the adversary
$\A=\{\{p_1\},\{p_2,p_3\},\{p_1,p_2,p_3\}\}$.
It is easy to see that $\setcon(\A)=2$, but that
$\setcon(\A|_{\Pi,\{p_2,p_3\}})= 1$ which is strictly smaller than
$\min(|\{p_2,p_3\}|,\setcon(\A))= 2$.
Therefore, $\A$ is non-{\fair}.

Consider the task $Cons_{2,3}$ consisting in consensus among $p_2$ and
$p_3$: every process in $\{p_2,p_3\}$ proposes a value and every
correct process in $\{p_2,p_3\}$ decides a proposed value, so that $p_2$ and $p_3$ cannot decide
different values.
%
$Cons_{2,3}$ is solvable in the adversarial $\A$-model:
every process in $\{p_2,p_3\}$ simply waits until $p_2$ writes its proposed value and decides on
it. Indeed, this protocol solves $Cons_{2,3}$ in the $\A$-model as if
$p_3$ is correct, $p_2$ is also correct.

The agreement function of $\A$, $\alpha_{\A}$, is equal to $0$ for
$\{p_2\}$ or $\{p_3\}$, to $2$ for $\{p_1,p_2,p_3\}$, and to $1$ for
all other values.
It is easy to see that $\alpha_{\A}$ only differs from
$\alpha_{1-res}$, the agreement function of the $1$-resilient
adversary, for $\{p_1\}$ where
$\alpha_{\A}(\{p_1\})=1>\alpha_{1-res}(\{p_1\})=0$.
Therefore, $\forall P\subseteq\Pi, \alpha_{\A}(P)\geq \alpha_{1-res}(P)$, and thus any task solvable in the $\A$-model is solvable in the $1$-resilient model.  

The impossibility of solving such a task $1$-resiliently can be
directly derived from the characterization of task solvable
$t$-resiliently from~\cite{Gaf09-EBG}.
Indeed, let $p_1$ wait for some
process to output in order to decide the same value. Processes $p_2$
and $p_3$ use the ability to solve consensus among themselves to
output a unique value. As there are two correct processes in the
system, $p_2$ or $p_3$ will eventually terminate and thus $p_1$ will
not wait indefinitely. This gives a 
$3$-process $1$-resilient consensus algorithm---a
contradiction\cite{FLP85,LA87}.
Thus, the $\A$-model is not equivalent with the
$\alpha_{\A}$-model, even though they have the same agreement function.

\section{Related work}
\label{sec:related}
Adversarial models were introduced by Delporte et
al. in~\cite{DFGT11}.
With respect to colorless tasks, Herlihy and Rajsbaum~\cite{HR12}
characterized a class \emph{superset-closed}~\cite{Kuz12}
adversaries (closed under the superset operation) via their minimal
core sizes. 
Still with respect to colorless tasks, Gafni and Kuznetsov~\cite{GK10}
derived a characterization of general adversary using its
\emph{consensus power} function $\setcon$.  
A side result of this present paper is an extension of the characterization
in~\cite{GK10} to any (not necessarily colorless) tasks.

Taubenfeld introduced in~\cite{Tau10} the notion of a symmetric progress
conditions, which is equivalent to our symmetric adversaries. 

%
The BG simulation establishes equivalence between $t$-resilience and wait-freedom with respect to task
solvability~\cite{BG93a,BGLR01,Gaf09-EBG}. Gafni and Guerraoui~\cite{GG11-univ} showed that if a model allows
for solving $k$-set consensus, then it can be used to simulate a \emph{$k$-concurrent} system in which at
most $k$ processes are concurrently invoking a task.
In our simulation, we use the fact that a model $M$ associated to an agreement function $\alpha_M$
allows to solve an $\alpha$-adaptive set consensus, using the technique proposed in~\cite{DFGK15},
which enables a composition of the ideas in~\cite{BG93a,BGLR01,Gaf09-EBG} and \cite{GG11-univ}. 
Running BG simulation on top of a $k$-concurrent system,
we are able to derive the equivalence between fair adversaries 
and their corresponding $\alpha$-models.


\section{Concluding remarks}
\label{sec:conc}

By Theorem \ref{th:adv:task}, task
computability
of a  {\fair} adversary $\A$ is \emph{characterized} by its 
agreement function $\alpha$: a task is solvable with $\A$ if and only
if it is solvable in the $\alpha$-model.
The result implies characterizations of superset-closed~\cite{HR10,Kuz12} and
symmetric~\cite{Tau10} adversaries and, via the equivalence
result established in~\cite{GG09}, the model of $k$-concurrency.

As a  corollary, for all models $M$ and $M'$ characterized by their
agreements functions, such that $\forall P\in
\Pi,\alpha_{M'}(P)\geq\alpha_M(P)$,
we have that $M$ is \emph{stronger} than $M'$, i.e., 
the set of tasks solvable in $M$ contains the set of tasks solvable in
$M'$.
In particular, if the two agreement functions are equal, then $M$ and
$M'$ solve exactly the same sets of tasks.
Note that if  a model $M$ is characterized by its agreement function $\alpha$, then it
belongs to the weakest equivalence class among the models whose
agreement function is $\alpha$.

An intriguing open question is therefore how to precisely determine the scope of the
approach based on agreement functions and if it can be extended to capture larger classes of models.

\bibliographystyle{abbrv}

\begin{thebibliography}{10}

\bibitem{AADGMS93}
Y.~Afek, H.~Attiya, D.~Dolev, E.~Gafni, M.~Merritt, and N.~Shavit.
\newblock Atomic snapshots of shared memory.
\newblock {\em J.~ACM}, 40(4):873--890, 1993.

\bibitem{BG93b}
E.~Borowsky and E.~Gafni.
\newblock Generalized {FLP} impossibility result for $t$-resilient asynchronous
  computations.
\newblock In {\em STOC}, pages 91--100. ACM Press, May 1993.

\bibitem{BG93a}
E.~Borowsky and E.~Gafni.
\newblock Immediate atomic snapshots and fast renaming.
\newblock In {\em PODC}, pages 41--51, New York, NY, USA, 1993. ACM Press.

\bibitem{BGLR01}
E.~Borowsky, E.~Gafni, N.~A. Lynch, and S.~Rajsbaum.
\newblock The {BG} distributed simulation algorithm.
\newblock {\em Distributed Computing}, 14(3):127--146, 2001.

\bibitem{DFGK15}
C.~Delporte{-}Gallet, H.~Fauconnier, E.~Gafni, and P.~Kuznetsov.
\newblock Wait-freedom with advice.
\newblock {\em Distributed Computing}, 28(1):3--19, 2015.

\bibitem{DFGT11}
C.~Delporte-Gallet, H.~Fauconnier, R.~Guerraoui, and A.~Tielmann.
\newblock The disagreement power of an adversary.
\newblock {\em Distributed Computing}, 24(3-4):137--147, 2011.

\bibitem{FLP85}
M.~J. Fischer, N.~A. Lynch, and M.~S. Paterson.
\newblock Impossibility of distributed consensus with one faulty process.
\newblock {\em J.~ACM}, 32(2):374--382, Apr. 1985.

\bibitem{Gaf09-EBG}
E.~Gafni.
\newblock The extended {BG}-simulation and the characterization of
  t-resiliency.
\newblock In {\em STOC}, pages 85--92, 2009.

\bibitem{GG09}
E.~Gafni and R.~Guerraoui.
\newblock Simulating few by many: Limited concurrency = set consensus.
\newblock Technical report, 2009.\\
\newblock \url{http://web.cs.ucla.edu/~eli/eli/kconc.pdf}.

\bibitem{GG11-univ}
E.~Gafni and R.~Guerraoui.
\newblock Generalized universality.
\newblock In {\em CONCUR}, pages 17--27, 2011.

\bibitem{GHKR16}
E.~Gafni, Y.~He, P.~Kuznetsov, and T.~Rieutord.
\newblock Read-write memory and k-set consensus as an affine task.
\newblock In {\em OPODIS}, 2016.
\newblock Technical report: \\
\newblock \url{https://arxiv.org/abs/1610.01423}.

\bibitem{GK10}
E.~Gafni and P.~Kuznetsov.
\newblock Turning adversaries into friends: Simplified, made constructive, and
  extended.
\newblock In {\em OPODIS}, pages 380--394, 2010.

\bibitem{GK11}
E.~Gafni and P.~Kuznetsov.
\newblock {Relating {\it L}-Resilience and Wait-Freedom via Hitting Sets}.
\newblock In {\em ICDCN}, pages 191--202, 2011.

\bibitem{Her91}
M.~Herlihy.
\newblock Wait-free synchronization.
\newblock {\em ACM Trans. Prog. Lang. Syst.}, 13(1):123--149, Jan. 1991.

\bibitem{HKR14}
M.~Herlihy, D.~N. Kozlov, and S.~Rajsbaum.
\newblock {\em Distributed Computing Through Combinatorial Topology}.
\newblock Morgan Kaufmann, 2014.

\bibitem{HR10}
M.~Herlihy and S.~Rajsbaum.
\newblock The topology of shared-memory adversaries.
\newblock In {\em PODC}, pages 105--113, 2010.

\bibitem{HR12}
M.~Herlihy and S.~Rajsbaum.
\newblock Simulations and reductions for colorless tasks.
\newblock In {\em PODC}, pages 253--260, 2012.

\bibitem{HS99}
M.~Herlihy and N.~Shavit.
\newblock The topological structure of asynchronous computability.
\newblock {\em J.~ACM}, 46(2):858--923, 1999.

\bibitem{Kuz12}
P.~Kuznetsov.
\newblock Understanding non-uniform failure models.
\newblock {\em Bulletin of the {EATCS}}, 106:53--77, 2012.

\bibitem{LA87}
M.~C. Loui and H.~H. Abu-Amara.
\newblock Memory requirements for agreement among unreliable asynchronous
  processes.
\newblock {\em advances in Computer Research}, 4:163--183, 1987.

\bibitem{SZ00}
M.~Saks and F.~Zaharoglou.
\newblock Wait-free k-set agreement is impossible: The topology of public
  knowledge.
\newblock {\em SIAM J. on Computing}, 29:1449--1483, 2000.

\bibitem{Tau10}
G.~Taubenfeld.
\newblock The computational structure of progress conditions.
\newblock In {\em DISC}, 2010.

\end{thebibliography}

\def\noopsort#1{} \def\No{\kern-.25em\lower.2ex\hbox{\char'27}}
  \def\no#1{\relax} \def\http#1{{\\{\small\tt
  http://www-litp.ibp.fr:80/{$\sim$}#1}}}

\end{document}